\documentclass[11pt]{article}
\usepackage[top=1in, bottom=1in, left=1in, right=1in]{geometry}

\usepackage[utf8]{inputenc} 
\usepackage[T1]{fontenc}    
\usepackage{hyperref}       
\usepackage{url}            
\usepackage{booktabs}       
\usepackage{amsfonts}       
\usepackage{nicefrac}       
\usepackage{microtype}      
\usepackage{graphicx}
\usepackage{subfig}      
\usepackage{algorithm}
\usepackage{algpseudocode}
\usepackage{multirow}
\usepackage{array, makecell}
\usepackage{amssymb}
\usepackage{amsmath}
\usepackage{amsthm}
\usepackage{arydshln}
\usepackage{cleveref}
\usepackage[numbers]{natbib}
\usepackage{parskip}

\newcommand\R{\mathbb R}
\renewcommand\P{\mathbb P}

\newtheorem{theorem}{Theorem}[section]

\newtheorem{definition}{Definition}[section]
\newtheorem{lemma}[theorem]{Lemma}

\newtheorem{problem}[theorem]{Problem}
\newtheorem{assumption}[theorem]{Assumption}

\newcommand{\eqlabel}[1]{\label{#1}\tag{#1}}

\newcommand\E{\mathbb{E}}

\newcommand{\norm}[1]{\ensuremath{\left\lVert #1 \right\rVert}}
\renewcommand{\norm}[1]{\lVert #1 \rVert}

\def\h_#1{\hat{#1}}
\def\wh_#1{\widehat{#1}}

\makeatletter
\newcommand\footnoteref[1]{\protected@xdef\@thefnmark{\ref{#1}}\@footnotemark}
\makeatother

\newcommand\set[1]{\left\{#1\right\}} 

\newcommand{\sm}{\text{sim}}

\newcommand{\pd}{\textrm{p}_{\textrm{data}}}
\newcommand{\pc}{\textrm{p}_{\textrm{control}}}
\newcommand{\ms}{\mu_{\textrm{same}}}
\newcommand{\md}{\mu_{\textrm{diff}}}

\title{Auditing for Diversity using Representative Examples\footnote{Accepted for publication at ACM-SIGKDD 2021.}}

\usepackage{authblk}

\author{Vijay Keswani}
\author{L. Elisa Celis}
\affil{Yale University}
\date{}

\begin{document}

\maketitle

\begin{abstract}
Assessing 
the diversity of a dataset of information associated with people is crucial before using such data for downstream applications.
For a given dataset, this often involves computing the \textit{imbalance} or \textit{disparity} in the empirical marginal distribution of a protected attribute (e.g. gender, dialect, etc.).
However, real-world datasets, such as images from Google Search or collections of Twitter posts, often do not have protected attributes labeled.
Consequently, to derive disparity measures for such datasets, the elements need to hand-labeled or crowd-annotated,
which are expensive processes.

We propose a cost-effective approach to approximate the disparity of a given unlabeled dataset, with respect to a protected attribute, using a control set of labeled representative examples.
Our proposed algorithm uses the pairwise similarity between elements in the dataset and elements in the control set to effectively
bootstrap 
an approximation to the disparity of the dataset.
Importantly, we show that using a control set whose size is much smaller than the size of the dataset is sufficient to achieve a small approximation error.
Further, based on our theoretical framework, we also provide an algorithm to construct \textit{adaptive} control sets that achieve smaller approximation errors than randomly chosen control sets.
Simulations on two image datasets and one Twitter dataset demonstrate the efficacy of our approach (using random and adaptive control sets) in auditing the diversity of a wide variety of datasets.

\end{abstract}

\newpage
\tableofcontents
\newpage


\section{Introduction}

Mechanisms to audit the \textit{diversity} of a dataset are necessary to assess the shortcomings of the dataset in representing the underlying distribution accurately.
In particular, any dataset containing information about people should suitably represent all social groups (defined by attributes such as gender, race, age, etc.) present in the underlying population in order to mitigate disparate outcomes and impacts in downstream applications \cite{buolamwini2018gender,caliskan2017semantics}.
However, many real-world and popular data sources suffer from the problem of disproportionate representation of minority groups \cite{noble2018algorithms,o2016weapons}.
For example, prior work has shown that the top results in Google Image Search for occupations are more gender-biased than the ground truth of the gender distribution in that occupation \cite{kay2015unequal,singh2019female,celis2020implicit}.

Given the existence of biased data collections in mainstream media and web sources, methods to audit the diversity of generic data collections can help quantify and mitigate the existing biases in multiple ways.
First, it gives a baseline idea of the demographic distribution in the collection and its deviation from the true distribution of the underlying population.
Second, stereotypically-biased representation of a social group in any data collection can lead to further propagation of negative stereotypes associated with the group \cite{harris1982mammies, word1974nonverbal, collins2002black} and/or induce incorrect perceptions about the group \cite{shrum1995assessing, gerbner1986living}.
A concrete example is the evidence of stereotype-propagation via  
biased Google Search results \cite{kay2015unequal, noble2018algorithms}.
These stereotypes and biases can be further exacerbated via machine learning models trained on the biased collections \cite{buolamwini2018gender, caliskan2017semantics,o2016weapons}.
Providing an easy way to audit the diversity in these collections can help the users of such collections assess the potential drawbacks and pitfalls of employing them for downstream applications.

Auditing the diversity of any collection with respect to a protected attribute primarily involves looking at the \textit{disparity} or \textit{imbalance} in the empirical marginal distribution of the collection with respect to the protected attribute.
For example, from prior work \cite{celis2020implicit}, we know that the top 100 Google Image Search results for CEOs in 2019 contained around 89 images of men and 11 images of women; in this case, we can quantify the disparity in this dataset, with respect to gender, as the difference between the fraction of minority group images and the fraction of majority group images, i.e., as $0.11{-}0.89{=}-0.78$.
The sign points to the direction of the disparity while the absolute value quantifies the extent of the disparity in the collection.
Now suppose that, instead of just 100 images, we had multiple collections with thousands of query-specific images, as in the case of Google Image Search. 
Since these images have been scraped or generated from different websites, the protected attributes of the people in the images will likely not be labeled at the source.
In the absence of protected attribute information, the task of simply auditing the diversity of these large collections (as an end-user) becomes quite labor-intensive.
Hand-labeling large collections can be extremely time-expensive, while using crowd-annotation tools (e.g. Mechanical Turk) can be very costly.
For a single collection, labeling a small subset (sampled i.i.d. from the collection) can be a reasonable approach to approximate the disparity; however, for multiple collections, this method is still quite expensive since, for every new collection, we will have to re-sample and label a new subset. It also does not support the addition/removal of elements to the collection.
One can, alternately, use automated models to infer the protected attributes; although, for most real-world applications, these supervised models need to be trained on large labeled datasets (which may not be available) and pre-trained models might encode their own pre-existing biases \cite{buolamwini2018gender}.

We, therefore, question if there is a cost-effective method to audit the diversity of large collections from a domain when the protected attribute labels of elements in the collections are unknown.

\noindent
\subsection{Our contributions}
The primary contribution of this paper is an algorithm to evaluate the diversity of a given unlabeled collection with respect to any protected attribute (Section~\ref{sec:model}).
Our algorithm takes as input the collection to be audited, a small set of labeled representative elements, called the \textit{control set}, and a metric that quantifies the similarity between any given pair of elements.
Using the control set and the similarity metric, our algorithm returns a proxy score of disparity in the collection with respect to the protected attribute.
The same control set can be used for auditing the diversity of any collection from the same domain.

The control set and the similarity metric are the two pillars of our algorithm, and we theoretically show the dependence of the effectiveness of our framework on these components.
In particular, the proxy measure returned by our algorithm approximates the true disparity measure with high probability, with the approximation error depending on the size and quality of the control set, and the quality of the similarity metric.
The protected attributes of the elements of the control set are expected to be labeled; however, the primary advantage of our algorithm is that the size of the control set can
be much smaller than the size of the collection to achieve small approximation error (Section~\ref{sec:ppb_discussion}).
Empirical evaluations on the Pilots Parliamentary Benchmark (PPB) dataset \cite{buolamwini2018gender} show that our algorithm, using randomly chosen control sets and cosine similarity metric, can indeed provide a reasonable approximation of the underlying disparity in any given collection (Section~\ref{sec:ppb_random_control}).

To further reduce the approximation error, we propose an algorithm to construct \textit{adaptive} control sets (Section~\ref{sec:adaptive_control}).
Given a small labeled auxiliary dataset, our proposed control set construction algorithm 
selects the elements that can best differentiate between samples with the same protected attribute type and samples with different protected attribute types.
We further ensure that the elements in the chosen control set are \textit{non-redundant} and \textit{representative} of the underlying population.
Simulations on PPB dataset, CelebA dataset \cite{liu2015faceattributes} and TwitterAAE dataset \cite{blodgett2016demographic} show that using cosine similarity metric and adaptive control sets, we can effectively approximate the disparity in random and topic-specific collections, with respect to a given protected attribute (Section~\ref{sec:adaptive_expts}).

\noindent
\subsection{Related work}
With rising awareness around the existence and harms of machine and algorithmic biases, prior research has explored and quantified disparities in data collections from various domains.
When the dataset in consideration has labeled protected attributes, the task of quantifying the disparity is relatively straightforward.
For instance, \citet{davidson2019racial} demonstrate racial biases in automated offensive language detection by using datasets containing Twitter posts with dialects labeled by the authors or domain experts.
\citet{larrazabal2020gender} can similarly analyze the impact of gender-biased medical imaging datasets since the demographic information associated with the images are available at source. 
However, as mentioned earlier, protected attribute labels for elements in a collection may not be available, especially if the collection contains elements from different sources.

In the absence of protected attribute labels from the source, crowd-annotation is one way of obtaining these labels and auditing the dataset.
To measure the gender-disparity in Google Image Search results, \citet{kay2015unequal} crowd-annotated a small subset of images and compared the gender distribution in this small subset to the true gender distribution in the underlying population.
Other papers on diversity evaluation have likewise used a small labeled subset of elements \cite{bolukbasi2016man, schwemmer2020diagnosing} to derive inferences about larger collections.
As discussed earlier, the problem with this approach is that it assumes that the disparity in the small labeled subset is a good approximation of the disparity in the given collection.
This assumption does not hold when we want to estimate the diversity of a new/multiple collections from the same domain or when elements can be continuously added/removed from the collection.
Our method, instead, uses a given small labeled subset to approximate the disparity measure of any collection from the same domain.
{Semi-supervised learning also explores learning methods that combine labeled and unlabeled samples \cite{zhu2009introduction}. The labeled samples are used to train an initial learning model and the unlabeled samples are then employed to improve the model generalizability.
Our proposed algorithm has similarities with the semi-supervised self-training approach \cite{bair2013semi}, but is faster and more cost-efficient 
(Section~\ref{sec:ppb_discussion}).
}

Representative examples have been used for other bias-mitigation purposes in recent literature.
Control or reference sets have been used for gender and skintone-diverse image summarization \cite{celis2020implicit}, dialect-diverse Twitter summarization \cite{celis2021dialect}, and fair data generation \cite{choi2020fair}.
\citet{kallus2020assessing} also employ reference sets for bias assessments; they approximate the disparate impact of prediction models in the absence of protected attribute labels.
In comparison, our goal is to evaluate representational biases in a given collection.

\section{Notations} \label{sec:notations}

\noindent
Let $S := \set{x_j}_{j=1}^N$ denote the collection to be evaluated. Each element in the collection consists a $d$-dimensional feature vector $x$, from domain $\mathcal{X} \subseteq \R^d$.
Every element $j$ in $S$ also has a protected attribute, $z_j \in \set{0,1}$, associated with it; however, we will assume that the protected attributes of the elements in $S$ are unknown.
Let $S_i := \set{x_j, j \in [N] \mid z_j =i}$.
A measure of disparity in $S$ with respect to the protected attribute is $d(S) := |S_0|/|S| - |S_1|/|S|$, i.e., the difference between fraction of elements from group 0 and group 1. A dataset $S$ is considered to be \textit{diverse} with respect to the protected attribute if this measure is 0, {and high $|d(S)|$ implies low diversity in $S$.}
Our goal will be to estimate this value for any given collection\footnote{Our proposed method can be used for other metrics that estimate imbalance in distribution of protected attribute as well (such as $|S_0|/|S|$);
however, for the sake of simplicity, we will limit our analysis to $d(S)$ evaluation.}
\footnote{We present the model and analysis for binary protected attributes. 
To extend the framework for non-binary protected attributes with $k$ possible values, one can alternately 
define disparity as $\max_{i \in [k]}|S_i| - \min_{i \in [k]}|S_i|$. }
Let $\pd$ denote the underlying distribution of the collection $S$.

\noindent
\paragraph{Control Set.} Let $T$ denote the control set of size $m$, i.e., a small set of representative examples. Every element $T$ also also has a feature vector from domain $\mathcal{X}$ and a protected attribute associated with it. 
Let $T_i := \set{x_j, j \in [m] \mid z_j =i}$.
Importantly, the protected attributes of the elements in the control set are known and we will primarily employ control sets that have equal number of elements from both protected attribute groups, i.e., $|T_0| = |T_1|$.
The size of control set is also much smaller than the size of the collection being evaluated, i.e., $|T| \ll |S|$.
Let $\pc$ denote the underlying distribution of the control set $T$.

Throughout the paper, we will also use the notation $a \in b \pm c$ to denote that $a \in [b-c, b+c]$.
The problem we tackle in this paper is auditing the diversity of $S$ using $T$; it is formally stated below.
\begin{problem} \label{pr:audit}
Given a collection $S$ (with \textbf{unknown} protected attributes of elements) and a balanced control set $T$ (with \textbf{known} protected attributes of elements), can we use $T$ to approximate $d(S)$?
\end{problem}

\section{Model and Algorithm} \label{sec:model}

The main idea behind using the control set $T$ to solve Problem~\ref{pr:audit} is the following: for each element $x \in S$, we can use the partitions $T_0, T_1$ of the control set to check which partition is most \textit{similar} to $x$. If most elements in $S$ are \textit{similar} to $T_0$, then $S$ can be said to have more elements with protected attribute $z{=}0$ (similarly for $z{=}1$).
However, to employ this audit mechanism we need certain conditions on the \textit{relevance} of the control set $T$, as well as, a metric that can quantify the similarity of an element in $S$ to control set partitions $T_0, T_1$. 
We tackle each issue independently below.

\noindent
\subsection{Domain-relevance of the control set}
To ensure that the chosen control set is representative and relevant to the domain of the collection in question, we will need the following assumption.
\begin{assumption} \label{asm:approp_control_set}
For any $x \in \mathcal{X}$, $\pd(x{\mid}z) = \pc(x{\mid}z)$, for all $z \in \set{0,1}$.
\end{assumption}
\noindent
This assumption states that the elements of control set are from the have the same conditional distribution as the elements of the collection $S$.
It roots out settings where one would try to use non-representative control sets for diversity audits (e.g., full-body images of people to audit the diversity of a collection of portrait images).
{Note that despite similar conditional distributions, the control set and the collection can (and most often will) have different protected attribute marginal distributions.}

We will use the notation $p_z(x)$ to denote the conditional distribution of $x$ given $z$ in the rest of the document, i.e., $p_z(x) := \pd(x \mid z) = \pc(x \mid z)$.
Given a collection $S$, we will call a control set $T$ (with partitions $T_0$, $T_1$) \textit{domain-relevant} if the underlying distribution of $T$ satisfies Assumption~\ref{asm:approp_control_set}.

\noindent
\subsection{Similarity metrics}
Note that even though $p_z(x)$ is the same for both the control set and the collection, the distributions $p_0(x)$ and $p_1(x)$ can be very different from each other, and our aim we will be to design and use similarity metrics that can differentiate between elements from the two conditional distributions.

A general pairwise similarity matrix $\sm: \mathcal{X} \times \mathcal{X} \rightarrow \R_{\geq 0}$ takes as input two elements and returns a non-negative score of similarity between the elements; the higher the score, the more similar are the elements.
For our setting we need a similarity metric that can, \textit{on average}, differentiate between elements that have the same protected attribute type and elements that have different protected attribute types.
Formally, we define such a similarity metric as follows.
\begin{definition}[$\gamma$-similarity metric] \label{defn:similarity}
Suppose we are given a similarity metric $\sm: \mathcal{X} \times \mathcal{X} \rightarrow [0,1]$, such that
\[\E_{x_1, x_2 \sim p_z}\left[ \sm(x_1, x_2) \right] = \ms \text{ and } \]
\[\E_{x_1 \sim p_{z_1}, x_2 \sim p_{z_2}, z_1 \neq z_2}\left[ \sm(x_1, x_2) \right] = \md.\]
Then for $\gamma \geq 0$, we call $\sm$ a $\gamma$-similarity metric if $\ms - \md \geq \gamma.$
\end{definition}
\noindent
Note that the above definition is not very strict; we do not require $\sm(\cdot, \cdot)$ to return a large similarity score for every pair of elements with same protected attribute type or to return a small similarity score for every pair of elements with different protected attribute types. 
Rather $\sm(\cdot, \cdot)$, only \textit{in expectation}, should be able to differentiate between elements from same groups and elements from different groups.
In a later section, we show that cosine similarity metric indeed satisfies this condition for real-world datasets.

\subsection{Algorithm}
Suppose we are given a \textit{domain-relevant} control set $T$ that satisfies Assumption~\ref{asm:approp_control_set} (with partitions $T_0$ and $T_1$) and a $\gamma$-similarity metric $\sm(\cdot, \cdot)$.
With slight abuse of notation, for any element $x \in S$, let $\sm(x, T_i) = \frac{1}{|T_i|} \sum_{y \in T_i} \sm(x,y)$ and let $\sm(S,T_i) = \frac{1}{|S|}\sum_{x \in S} \sm(x, T_i)$.
Let $\hat{d}(S) := \sm(S, T_0) - \sm(S,T_1)$.
%
We propose the use of $\hat{d}(S)$ (after appropriate normalization) as a proxy measure for $d(S)$; Algorithm~\ref{alg:main} presents the complete details of this proxy diversity score computation and Section~\ref{sec:theory} provides bounds on the approximation error of $\hat{d}(S)$.
We will refer to the value returned by Algorithm~\ref{alg:main} as \textit{DivScore} for the rest of the paper.

\begin{algorithm} 
\caption{\textit{DivScore}: Algorithm for proxy diversity audit}
\label{alg:main}
\begin{flushleft}
    \textbf{Input:} Dataset $S$, control set $T := T_0 \cup T_1$, similarity metric $\sm(\cdot, \cdot)$
\end{flushleft}
\begin{algorithmic}[1]
\State $l\gets \frac{1}{|T_0|\cdot |T_1|} \sum_{x, y \in T_0 \times T_1} \sm(x,y)$
\State $u_0\gets \frac{1}{|T_0|\cdot (|T_0| - 1)} \sum_{x\in T_0, y\in T_0\setminus \set{x}} \sm(x,y)$
\State $u_1\gets \frac{1}{|T_1|\cdot (|T_1| - 1)} \sum_{x\in T_1, y\in T_1\setminus \set{x}} \sm(x,y)$
\State Compute $\sm(S,T_0) \gets \frac{1}{|S| \cdot|T_0|} \sum_{x,y \in S \times T_0} \sm(x,y)$
\State $s_0 \gets (\sm(S,T_0) - l)/(u_0-l)$
\State Compute $\sm(S,T_1) \gets \frac{1}{|S| \cdot|T_1|} \sum_{x,y \in S \times T_1} \sm(x,y)$
\State $s_1 \gets (\sm(S,T_1) - l)/(u_1-l)$
\State \textbf{return} $s_0 - s_1$ 
\end{algorithmic}
\end{algorithm}

\subsection{Theoretical analysis} \label{sec:theory}

To prove that $\hat{d}(S)$ is a good proxy measure for auditing diversity,
we first show that if $x \in S_i$, then $\sm(x,T_i) > \sm(x,T_j)$, for $j = 1-i$, with high probability and quantify the exact difference using the following lemma.
For the analysis in this section, assume that the elements in $T_0$, $T_1$ have been sampled i.i.d. from conditional distribution $p_0, p_1$ respectively and $|T_0| = |T_1|$.

\begin{lemma} \label{lem:single_elem}
For $i \in \set{0,1}$, any $x \in S_i$ and $\delta >0$, with probability atleast $1 - 2e^{-\delta^2 \md |T| /6} \cdot (1 + e^{-\delta^2 \gamma |T| /6} )$, we have
\begin{align*}
\sm(x, T_i) - \sm(x,T_{1-i}) \in  \ms - \md \pm \delta (\ms + \md). \eqlabel{1}
\end{align*}
\end{lemma}
\noindent
The lemma basically states that a $\gamma$-similarity metric, with high probability, can differentiate between $\sm(x, T_i)$ and $\sm(x, T_{1-i})$. 
The proof uses the fact that since $T$ is domain-relevant and the elements of $T$ are i.i.d. sampled from the conditional distributions, for any $y \in T_0$, $\E[\sm(x,y)] = \ms$ and for any $y \in T_1$, $\E[\sm(x,y)] = \md$.
Then, the statement of the lemma can be proven using standard Chernoff-Hoeffding concentration inequalities \cite{hoeffding1994probability,mitzenmacher2017probability}.
The complete proof is presented in Appendix~\ref{sec:proofs}.
Note that even though $\sm$ was defined to differentiate between protected attribute groups in expectation, by averaging over all control set elements in $T_0, T_1$, we are able to differentiate across groups with high probability.

The lemma also partially quantifies the dependence on $|T|$ and $\gamma$.
Increasing the size of control set $T$ will lead to higher success probability. Similarly, larger $\gamma$ implies that the similarity metric is more powerful in differentiating between the groups, which also leads to higher success probability.
Using the above lemma, we can next prove that the proposed diversity audit measure is indeed a good approximation of the disparity in $S$.
Recall that, for the dataset $S$, $\sm(S,T_i) = \frac{1}{|S|}\sum_{x \in S} \sm(x, T_i)$.

\begin{theorem}[Diversity audit measure] \label{thm:main}
For protected attribute $z \in \set{0,1}$, let $p_z$ denote the underlying conditional distribution $\pd(x{\mid}z)$.
Suppose we are a given a dataset $S$ containing i.i.d. samples from $\pd$, a domain-relevant control set $T$ (with pre-defined partitions by protected attribute $T_0$ and $T_1$, such that $|T_0| = |T_1|$) and a similarity metric $\sm : \mathcal{X}^2 \rightarrow \R_{\geq 0}$, such that if \\$\ms = \E_{x_0, x_1 \sim p_z}\left[ \sm(x_0, x_1) \right],$ $\md = \E_{x_0 \sim p_0, x_1 \sim p_1 }\left[ \sm(x_0, x_1) \right],$
then $\ms - \md \geq \gamma$, for $\gamma > 0$.

\noindent
Let $\delta = \sqrt{\frac{6\log (20|S| )}{ |T| \min(\md, \gamma)}}$ and let $\hat{d}(S) := \sm(S, T_0) - \sm(S,T_1)$.
Then,
with high probability, 
$\hat{d}(S)/(\ms - \md)$ approximates $d(S)$ within an additive error of $ \delta \cdot (\ms + \md)/(\ms - \md)$.

\noindent
In particular, with probability $\gtrapprox$ 0.9,
\[\hat{d}(S) \in(\ms - \md) \cdot  d(S) \pm \delta \cdot (\ms + \md).\]
\end{theorem}

\noindent
Theorem~\ref{thm:main} states that, with high probability, $d(S)$ is contained in a small range of values determined by $\hat{d}(S)$, i.e., $$d(S) \in \left( \hat{d}(S) \pm \delta \cdot (\ms + \md) \right) / \left( \ms - \md) \right).$$ 
The theoretical analysis is line with the implementation in Algorithm~\ref{alg:main} (\textit{DivScore}), i.e., the algorithm computes $\hat{d}(S)$ and normalizes it appropriately using estimates of $\ms$ and $\md$ derived from the control set.
The proof is presented in Appendix~\ref{sec:proofs_2}.

Note that Theorem~\ref{thm:main} assumes that $\ms{=}\E_{x_0, x_1 \sim p_z}\left[ \sm(x_0, x_1) \right]$ is the same for both $z \in \set{0,1}$.
However, they may not be same in practice and \textit{DivScore} uses separate upper bounds for $z{=}0$ and $z{=}1$ ($u_0$ and $u_1$ respectively).
Similarly, we don't necessarily require a balanced control set (although, as discussed in Section~\ref{sec:ppb_discussion}, a balanced control set is preferable over an imbalanced one).
We keep the theoretical analysis simple for clarity, but both these changes can be incorporated in Theorem~\ref{thm:main} to derive similar bounds as well.

The dependence of error on $\gamma$ and $T$ can also be inferred from Theorem~\ref{thm:main}.
The denominator in the error term in Theorem~\ref{thm:main} is lower bounded by $\gamma$. Therefore, the larger the $\gamma$, the lower is the error and the tighter is the bound.
The theorem also gives us an idea of the size of control set required to achieve low $\delta$ error and high success probability. 
To keep $\delta$ small, we can choose a control set $T$ with $|T| = \Omega\left(\log |S|\right)$. In other words, a control set of size $c\log |S|$ elements, for an appropriate $c>1$, should be sufficient to obtain low approximation error.
Since the control sets are expected to have protected attribute labels (to construct partitions $T_0$ and $T_1$), having small control sets will make the usage of our audit algorithm much more tractable.

\noindent
\paragraph{Cost of \textit{DivScore}.}
The time complexity of Algorithm~\ref{alg:main} (\textit{DivScore}) is $O(|S|{\cdot}|T|)$, and it only requires $|T|$ samples (control set) to be labeled.
In comparison, if one was to label the entire collection to derive $d(S)$, the time complexity would be $O(|S|)$, but all $|S|$ samples would need to be labeled.
With a control set $T$ of size $\Omega\left(\log |S|\right)$, our approach is much more cost-effective.
The elements of $T$ are also not dependent on elements of $S$; hence, the same control set can be used for other collections from the same domain.

\begin{figure*}
    \centering
    \includegraphics[width=0.9\linewidth]{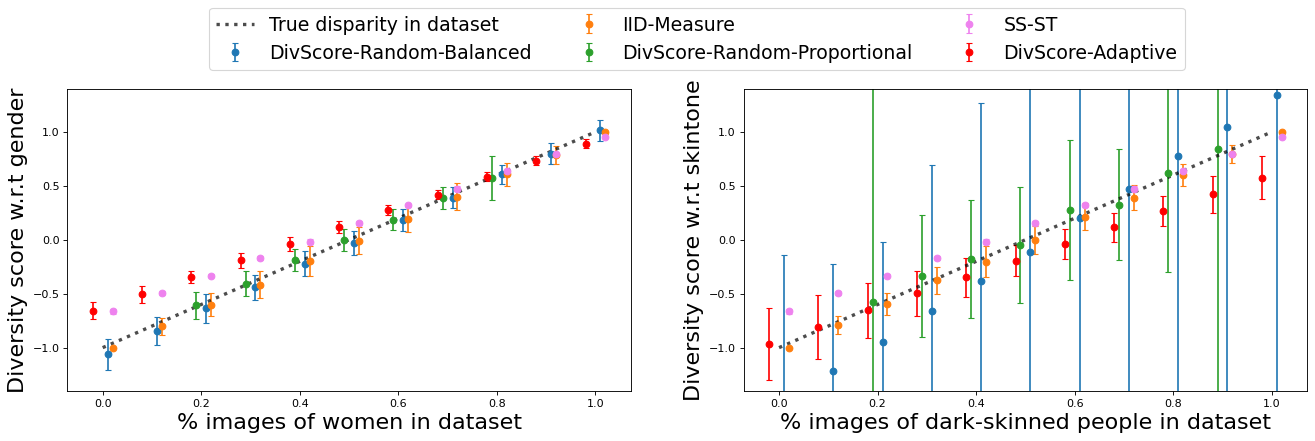}
    \subfloat[Gender protected attribute]{\hspace{.5\linewidth}}
    \subfloat[Skin-tone protected attribute]{\hspace{.5\linewidth}}
    \caption{\small Results for PPB-2017 dataset using random and adaptive control sets. The reported performance is the mean of output from \textit{DivScore} across 100 repetitions (errorbars denote standard error).
    To improve readability, we limit the y-axis to range to $[-1.5,1.5]$, which results in trimmed errorbands for some methods; {we present the same expanded plots without axis restrictions in Appendix~\ref{sec:other_results}}.
    The protected attribute considered here are gender and skintone.
    The x-axis reports the fraction of $z=0$ images in the collection ($\in \set{0, 0.1, 0.2, \dots, 1.0}$) and, for each collection, we report the following five metrics in y-axis: true disparity of the collection, \textit{DivScore-Random-Balanced}, \textit{DivScore-Random-Proportional}, \textit{IID-Measure}, and \textit{DivScore-Adaptive}.
   {A collection is considered diverse if the diversity score (y-axis) is 0; the larger the deviation of diversity score from 0, the lower the diversity is in the evaluated collection.}
    Amongst all metrics, \textit{DivScore-Adaptive}, \textit{IID-Measure}, and \textit{SS-ST} seem to have the lowest standard error.
    However, using \textit{IID-Measure} and \textit{SS-ST} are much costlier than \textit{DivScore}, as discussed in Section~\ref{sec:ppb_discussion}.
    }
    \label{fig:ppb_results_all}
\end{figure*}

\section{Empirical Evaluation Using Random Control Sets} \label{sec:experiments}

We first demonstrate the efficacy of the \textit{DivScore} algorithm on a real-world dataset using random, \textit{domain-relevant} control sets.

\subsection{PPB-2017 dataset} \label{sec:ppb_random_control}
The PPB (Pilots Parliamentary Benchmark) dataset consists of 1270 portrait images of parliamentarians from six different countries\footnote{\url{gendershades.org}}. The images in this dataset are labeled with gender (male vs female) and skin-tone (values are the 6 types from the Fitzpatrick skin-type scale \cite{fitzpatrick2008fitzpatrick}) of the person in the image.
This dataset was constructed and curated by \citet{buolamwini2018gender}.
We will use gender and skintone as the protected attributes for our diversity audit analysis.

\noindent
\paragraph{Methodology.}
We first the split the dataset into two parts: the first containing 200 images and the second containing 1070 images.
The first partition is used to construct control sets,
while the second partition is used for diversity audit evaluation.
Since we have the gender and skin-tone labels for all images, we can construct sub-datasets of size 500 with custom distribution of protected attribute types.
In other words, for a given $f \in \set{0, 0.1, 0.2, \dots, 1.0}$, we construct a sub-dataset $S$ of the second partition containing $f\cdot|S|$ images corresponding to protected attribute $z=0$.
Hence, by applying Algorithm~\ref{alg:main} (\textit{DivScore}) using a given control set $T$, we can assess the performance of our proxy measure for collection with varying fraction of under/over-represented group elements.

When protected attribute is gender, $z=0$ will denote $g = \text{female}$, when protected attribute is skin-tone, $z=0$ will denote $s > 3$ (skin-tone types corresponding to dark-skin), and when protected attribute is intersection of gender and skin-tone, $z=0$ will denote $g = \text{female and } s > 3$ (corresponding to dark-skinned women).

\noindent
\paragraph{Control sets.} 
{To evaluate the performance of \textit{DivScore}
the selection of elements for the control sets (of size 50 from the first partition) can be done in multiple ways: 
(1) \textit{random balanced control sets}, i.e., randomly block-sampled control sets with equal number of $z{=}0$ and $z{=}1$ images; (2) \textit{random proportional control sets}, i.e., control sets i.i.d. sampled from the collection in question; (3) \textit{adaptive control sets}, i.e., non-redundant control sets that can best differentiate between samples with the same protected attribute type and samples with different protected attribute types. 
The complete details of construction of \textit{adaptive control sets} is given in Section~\ref{sec:adaptive_control}; in this section, we primarily focus on performance of \textit{DivScore} when using random control sets.
We will refer to our method as \textit{DivScore-Random-Balanced}, when using random balanced control sets, and as \textit{DivScore-Random-Proportional}, when using random proportional control sets.}
In expectation, random proportional control sets will have a similar empirical marginal distribution of protected attribute types as the collection; correspondingly, we also report the disparity measure of the random proportional control set $d(T)$ as a baseline.
We will refer to this baseline as \textit{IID-Measure}.
Random proportional control sets need to be separately constructed for each new collection, while the same random balanced control set can be used for all collections; we discuss this contrast further in Section~\ref{sec:ppb_discussion}.

{We also implement a semi-supervised self-training algorithm as a baseline. This algorithm (described formally in Appendix~\ref{sec:impl_details}) iteratively labels the protected attribute of those elements in the dataset
for which similarity to one group in the control set is significantly larger than similarity to the other group. It then uses the learnt labels to compute the diversity score. We implement this baseline using random control sets and refer to it as \textit{SS-ST}.
\footnote{{We do not compare against crowd-annotation since the papers providing crowd-annotated datasets in our considered setting usually do not have ground truth available to estimate the approximation error.}}
}

\noindent
\paragraph{Similarity Metric.}
We construct feature vector representations for all images in the dataset using pre-trained deep image networks.
The feature extraction details are presented in Appendix~\ref{sec:impl_details}.
Given the feature vectors, we use the cosine similarity metric to compute pairwise similarity between images.
In particular, given feature vectors $x_1, x_2$ corresponding to any two images, we will define the similarity between the elements as 
\[\setlength{\abovedisplayskip}{3pt}
\setlength{\belowdisplayskip}{3pt} \sm(x_1, x_2) = 1 + \frac{x_1^\top x_2}{\norm{x_1}\norm{x_2}}. \eqlabel{2}\]
We add 1 to the standard cosine between two vectors to ensure that the similarity values are always non-negative.

\noindent
\paragraph{Evaluation Measures.}
We repeat the simulation 100 times; for each repetition, we construct a new split of the dataset and sample a new control set.
We report the true fraction $f$ and the mean (and standard error) of 
all metrics
across all repetitions.

\noindent
\paragraph{Results.}
The results are presented in Figure~\ref{fig:ppb_results_all} (the figure also plots the performance of \textit{DivScore-Adaptive}, that is discussed in Section~\ref{sec:adaptive_control}).
With respect to gender, Figure~\ref{fig:ppb_results_all}a shows that the \textit{DivScore} measure is always close to the true disparity measure for all collections, and the standard error of all metrics is quite low.
In this case, random control sets (balanced or proportional) can indeed approximate the disparity of all collections with very small error.

The results are more mixed when skintone is the protected attribute.
Figure~\ref{fig:ppb_results_all}b shows that while the \textit{DivScore} average is close to the true disparity measure, the standard errors are quite high.
The baselines \textit{IID-Measure} and \textit{SS-ST} have lower errors than our proxy measure (although they are not a feasible method for real-world applications, as discussed in the next section).
The poor performance for this protected attribute, when using random control sets, suggests that strategies to construct \textit{good} non-random control sets are necessary to reduce the approximation error.

\subsection{Discussion} \label{sec:ppb_discussion}
\noindent
\paragraph{Drawbacks of \textit{IID-Measure}.} Recall that \textit{IID-Measure} essentially uniformly samples a small subset of elements of the collection and reports the disparity of this small subset.
Figure~\ref{fig:ppb_results_all} shows that this baseline indeed performs well for PPB-dataset. However, it is not a cost-effective approach for real-world disparity audit applications.
The main drawback of this baseline is that the subset has to have i.i.d. elements from the collection being audited for it to accurately predict the disparity of the collection.
This implies that, for every new collection, we will have to re-sample and label a small subset to audit its diversity using \textit{IID-Measure}.
It is unreasonable to apply this approach when there are multiple collections (from the same domain) that need to be audited or when elements are continuously being added/removed from the collection.
The same reasoning limits the applicability of \textit{DivScore-Random-Proportional}.

\textit{DivScore-Random-Balanced}, on the other hand, addresses this drawback by using a generic labeled control set that can be used for any collection from the same domain, without additional overhead of constructing a new control set everytime.
This is also why balanced control sets should be preferred over imbalanced control sets, since a balanced control set will be more adept at handling collections with varying protected attribute marginal distributions. 

\noindent
\paragraph{Drawbacks of \textit{SS-ST}.}
{The semi-supervised learning baseline \textit{SS-ST} has larger estimation bias than \textit{DivScore-Random-Balanced} and \textit{DivScore-Random-Proportional}, but has lower approximation error than these methods.
However, the main drawback of this baseline is the time complexity. Since it iteratively labels elements and then adds them to control set to use for future iterations, the time complexity of this baseline is quadratic in dataset size.
In comparison, the time complexity of \textit{DivScore} is linear in the dataset size.}

\noindent
\paragraph{Dependence on $\gamma$.}
The performance of \textit{DivScore}
on PPB-dataset highlights the dependence of approximation error on the $\gamma$.
Since the gender and skintone labels of images in the dataset are available, we can empirically derive the $\gamma$ value for each protected attribute using the cosine similarity metric.
When gender is the protected attribute, $\gamma$ is around 0.35.
On the other hand, when skintone is the protected attribute, $\gamma$ is 0.08.
In other words, the cosine similarity metric is able to differentiate between images of men and women to a better extent than between images of dark-skinned and light-skinned people.
This difference in $\gamma$ is the reason for the relatively larger error of \textit{DivScore} in case of skintone protected attribute. 

\noindent
\paragraph{Cosine similarity metric.}
The simulations also show that measuring similarity between images using the cosine similarity metric over feature vectors from pre-trained networks is indeed a reasonable strategy for disparity measurement.
Pre-trained image networks and cosine similarity metric has similarly also been used in prior work for classification and clustering purposes \cite{nguyen2010cosine, xing2003distance}.
Intuitively, cosine similarity metric is effective when conditional distributions $p_0$ and $p_1$ are concentrated over separate clusters over the feature space;
e.g., for PPB-dataset and gender protected attribute, the high value of $\gamma$ (0.35) provides evidence of this phenomenon.
In this case, cosine similarity can, \textit{on average}, differentiate between elements from same cluster and different clusters.

\noindent
\paragraph{Dependence on $|T|$.}The size of control set is another factor which is inversely related to the error of the proxy disparity measure.
For this section, we use control sets of size 50. 
Smaller control sets lead to larger variance, as seen in Figure~\ref{fig:ppb_diff_control} in the Appendix, while using larger control sets might be inhibitory and expensive since, in a real-world applications, protected attributes of the control set images need to be hand-labeled or crowd-annotated.

Nevertheless, these empirical results highlight the crucial dependence on $\gamma$ and properties of the control set $T$.
In the next section, we 
improve upon the performance of our disparity measure 
and reduce the approximation error by designing non-random control sets that can better differentiate across the protected attribute types.

\begin{algorithm} 
\caption{Algorithm to construct an \textit{adaptive} control set}
\label{alg:control_construction}
\begin{flushleft}
    \textbf{Input:} Auxiliary set $U = U_0{\cup}U_1$, similarity metric $\sm$, $m$, $\alpha \geq 0$
\end{flushleft}
\begin{algorithmic}[1]
\State $T_0, T_1, \gamma_0, \gamma_1 \gets \emptyset $
\For{$i \in \set{0,1}$}
    \For{$x \in U_i$}
    \State { $\gamma_i^{(x)}{\gets}\frac{1}{|U_i|-1} \sum\limits_{y \in U_i \setminus \set{x}} \sm(x,y){-}\frac{1}{|U_{1-i}|} \sum\limits_{y \in U_{1-i}} \sm(x,y)$ }
    \EndFor
    \While{$|T_i| < m/2$}
        \State { $ T_i \gets T_i \cup \set{\arg \max\set{\gamma_i^{(x)} - \alpha \cdot \max_{y \in T_i} \sm(x,y) }_{x \in U_i \setminus T_i}}$}
    \EndWhile
\EndFor
\State \textbf{return} $T_0 \cup T_1$ 
\end{algorithmic}
\end{algorithm}
\raggedbottom

\section{Adaptive Control Sets} \label{sec:adaptive_control}
The theoretical analysis in Section~\ref{sec:theory} and the simulations in Section~\ref{sec:ppb_random_control} use random control sets; i.e., $T$ contains i.i.d. samples from $p_0$ and $p_1$ conditional distributions.
This choice was partly necessary because the error depends on the $\gamma$-value of the similarity metric, which is quantified as $\ms - \md$, where
\[\ms{=}\E_{x_0, x_1 \sim p_z}\left[ \sm(x_0, x_1) \right], \; \; \; \md{=}\E_{x_0 \sim p_{0}, x_1 \sim p_{1}}\left[ \sm(x_0, x_1) \right].\]
However, quantifying $\ms, \md$ (and, hence, $\gamma$) using expectation over the entire distribution might be unnecessary.

In particular, the theoretical analysis uses $\ms$ to quantify $\E_{x \sim p_i}\left[ \sm(x, T_i) \right]$, for any $i \in \set{0,1}$ (similarly $\md$).
Hence, we require the difference between $\ms$ and $\md$ to be large only when comparing the elements from the underlying distribution to the elements in the control set.
This simple insight provides us a way to choose \textit{good} control sets; i.e., we can choose control sets $T$ for which the difference $|\E_x\left[ \sm(x, T_i) \right] - \E_x\left[ \sm(x, T_{1-i}) \right]|$ is 
large.

\noindent
\paragraph{Control sets that maximize $\gamma$.}
Suppose we are given an auxiliary set $U$ of i.i.d. samples from $\pd$, such that the protected attributes of elements in $U$ are known. Let $U_0, U_1$ denote the partitions with respect to the protected attribute.
Once again, $U \ll |S|$ 
and $U$ will be used to construct a control set $T$.
Let $m{\in}\set{0, 2, 4, \dots, |U|}$ denote the desired size of $T$.
For each $i{\in}\set{0,1}$ and $y{\in}U_i$, we can first compute
\[\gamma_i^{(y)} := \E_{x\sim U_i\setminus\set{y}}\left[ \sm(x, y) \right] - \E_{x \sim U_{1-i}\setminus\set{y}}\left[ \sm(x, y) \right], \]
and then construct a control set $T$ by
adding $m/2$ elements from each $U_i$ with the largest values in the set $\set{\gamma_i^{(y)}}_{y \in U_i}$ to $T$.

\begin{figure*}[t]
    \centering
    \includegraphics[width=\linewidth]{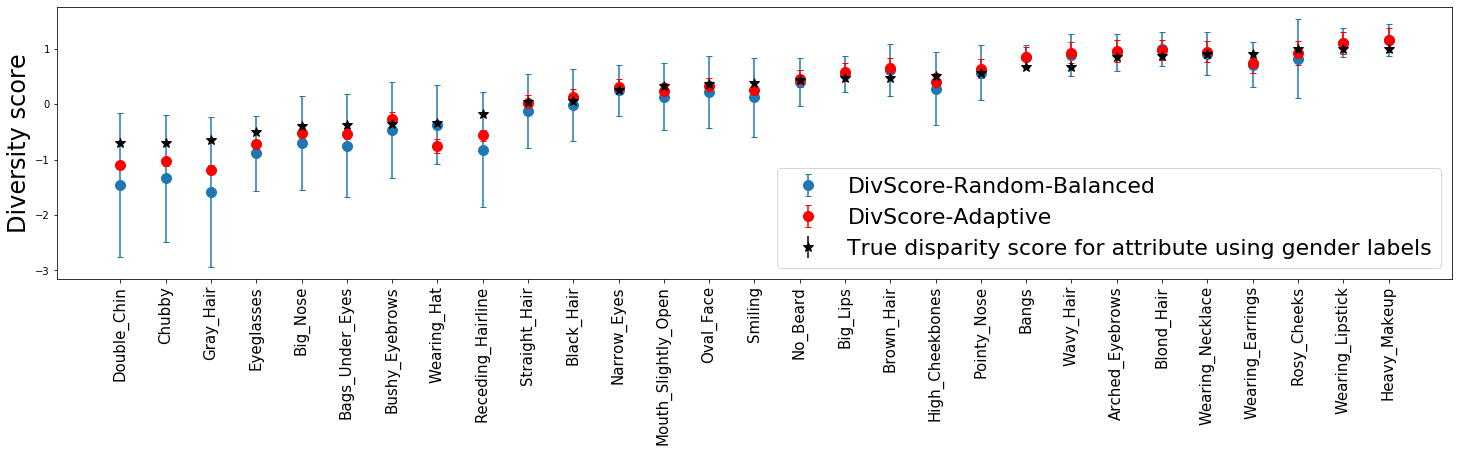}
    \caption{\small Results for CelebA dataset. For each feature, we plot the true gender disparity score for that feature as well as the scores obtained using \textit{DivScore-Random-Balanced} and \textit{DivScore-Adaptive} approaches. For both methods, the control set size is kept to 50. Note that the error of \textit{DivScore-Adaptive} is much smaller in this case.}
    \label{fig:celeba_results}
\end{figure*}
\raggedbottom

\noindent
\paragraph{Reducing redundancy in control sets.}
While the above methodology will result in control sets that maximize the difference between similarity with same group elements vs similarity with different group elements, it can also lead to \textit{redundancy} in the control set.
For instance, if two elements in $U$ are very similar to each other, they will large pairwise similarity and can, therefore, both have large $\gamma_i^{(y)}$ value 
; however, adding both to the control set is redundant.
Instead, we should aim to make the control set as \textit{diverse} and \textit{representative} of the underlying population as possible.
To that end, we employ a Maximal Marginal Relevance (MMR)-type approach and iteratively add elements from $U$ to the control set $T$.
For the first $m/2$ iterations, we add elements from $U_0$ to $T$.
Given a hyper-parameter $\alpha \geq 0$, at any iteration $t$, the element added to $T$ is the one that maximizes the following score:
$$ \set{ \gamma_0^{(y)} - \alpha  \cdot \max_{x \in T} \sm(x,y) }_{y \in U_0 \setminus T}. $$
The next $m/2$ iterations similarly adds elements from $U_1$ to $T$ using $\gamma_1^{(y)}$.
The quantity $\max_{x \in T} \sm(x,y)$ is the \textit{redundancy score} of $y$; i.e., the maximum similarity of $y$ with any element already added to $T$.
By penalizing an element for being very similar to an existing an element in $T$, we can ensure that chosen set $T$ is diverse.
The complete algorithm to construct such a control set, using a given $U$, is provided in Algorithm~\ref{alg:control_construction}.
We will refer to the control sets constructed using Algorithm~\ref{alg:control_construction} as \textit{adaptive} control sets and Algorithm~\ref{alg:main} with adaptive control sets as \textit{DivScore-Adaptive}.

Note that, even with this control set construction method, the theoretical analysis does not change.
Given any control set $T$ ($=T_0 \cup T_1$), let $\gamma^{(T)} := \E_{i} \left[ \E_{x \in p_i}\left[ \sm(x, T_i) \right] - \E_{x \sim p_{1-i}}\left[ \sm(x, T_{1-i}) \right] \right].$
For a control set $T$ with parameter $\gamma^{(T)}$, we can obtain the high probability bound in Theorem~\ref{thm:main} by simply replacing $\gamma$ by $\gamma^{(T)}$.
Infact, since we are explicitly choosing elements that have large $\gamma_i^{(\cdot)}$ parameters, 
$\gamma^{(T)}$ is expected to be larger than $\gamma$ and, hence, using the an adaptive control set will lead to a stronger bound in Theorem~\ref{thm:main}.

Our algorithm uses the standard MMR framework to reduce redundancy in the control set.
Importantly, prior work has shown that the greedy approach of selecting the \textit{best available} element is indeed approximately optimal \cite{carbonell1998use}.
Other non-redundancy approaches, e.g., Determinantal Point Processes \cite{kulesza2012determinantal}, can also be employed.

\noindent
\paragraph{Cost of each method.} \textit{DivScore-Adaptive} requires an auxiliary labeled set $U$ from which we extract a good control set. Since $|U| > |T|$, the cost (in terms of time and labeling required) of using \textit{DivScore-Adaptive} is slightly larger than the cost of using \textit{DivScore-Random-Balanced}, for which we just need to randomly sample $|T|$ elements to get a control set.
However, results in Appendix~\ref{sec:other_results} show that, to achieve similar approximation error, the required size of adaptive control sets is smaller than the size of random control sets.
Hence, even though adaptive control sets are more costly to construct, \textit{DivScore-Adaptive} is more cost-effective for disparity evaluations and requires smaller control sets (compared to \textit{DivScore-Random-Balanced}) to approximate with low error.

\section{Empirical Evaluation using adaptive control sets} \label{sec:adaptive_expts}

\subsection{PPB-2017}
Once again, we first test the performance of adaptive control sets on PPB-2017 dataset.
Recall that we split the dataset into two parts of size 200 and 1070 each.
Here, the first partition 
serves as the auxiliary set $U$ for Algorithm~\ref{alg:control_construction}.
The input hyper-parameter $\alpha$ is set to be 1.
The rest of the setup is the same as in Section~\ref{sec:ppb_random_control}.

\noindent
\paragraph{Results.} The results for this simulation are presented in Figure~\ref{fig:ppb_results_all} (in red).
The plots show that, using adaptive control sets, we obtain sharper proxy diversity measures for both gender and skintone.
For skintone protected attribute, the standard error of \textit{DivScore-Adaptive} is significantly lower than \textit{DivScore-Random-Balanced}.

Note that the average of \textit{DivScore-Adaptive}, across repetitions, does not align with the true disparity measure (unlike the results in case of random control sets).
This is because the {adaptive} control sets do not necessarily represent a uniformly random sample from the underlying conditional distributions.
Rather, they are the subset of images from $U$ with best scope of differentiating between images from different protected attributes types.
This non-random construction of the control sets leads to a possibly-biased but tighter approximation for the true disparity in the collection.

As noted before, when using adaptive control sets (from Algorithm~\ref{alg:control_construction}), the performance depends on the measure $ \gamma^{(T)} := \E_{i} \left[ \E_{x \in p_i}\left[ \sm(x, T_i) \right] - \E_{x \sim p_{1-i}}\left[ \sm(x, T_{1-i}) \right] \right].$
By construction, we want to choose control sets $T$ for which $\gamma^{(T)}$ is greater than the $\gamma$ value over the entire distribution.
Indeed, in case of PPB dataset and for every protected attribute, we observe that $\gamma^{(T)}$ values of the adaptive control sets are much larger than the corresponding value when of randomly chosen control sets.
When gender is the protected attribute, on average, $\gamma^{(T)}$ is 0.96 (for random control sets, it was 0.35).
Similarly, when skintone is the protected attribute, $\gamma^{(T)}$ is around 0.34 (for random control sets, it was 0.08).
The stark improvement in these values, compared to random control sets, is the reason behind the increased effectiveness of adaptive control sets in approximating the disparity of the collection.

\subsection{CelebA dataset}
CelebA dataset \cite{liu2015faceattributes} contains images of celebrities with tagged facial attributes, such as whether the person in the image has eyeglasses, mustache, etc., along with the gender of the person in the image\footnote{\url{mmlab.ie.cuhk.edu.hk/projects/CelebA.html}}.
We use 29 of these attributes and a random subset of around 20k images for our evaluation.
The goal is to approximate the disparity in the collection of images corresponding to a given facial attribute.

\noindent
\paragraph{Methodology.}
We evaluate the performance \textit{DivScore-Random-Balanced} and \textit{DivScore-Adaptive} for this dataset\footnote{\label{note1}{For CelebA and TwitterAAE datasets, we only report the performance of \textit{DivScore-Adaptive} and \textit{DivScore-Random-Balanced} to ensure that the plots are easily readable. 
The performance of \textit{DivScore-Random-Balanced} is similar to that of \textit{DivScore-Random-Balanced} and, due to large data collection sizes, \textit{SS-ST} is infeasible in this setting.}}.
We perform 25 repetitions; in each repetition, an auxiliary set $U$ is sampled of size 500 (and removed from the main dataset) and used to construct either a random control set (of size 50) or an adaptive control set (of size 50).
The chosen control set is kept to be the same for all attribute-specific collections in a repetition.
For each image, we use the pre-trained image networks to extract feature vectors (see Appendix~\ref{sec:impl_details} for details) and the cosine similarity metric - Eqn~$\eqref{2}$ - to compute pairwise similarity.

\noindent
\paragraph{Results.}
The results are presented in Figure~\ref{fig:celeba_results}.
The plot shows that, for almost all attributes, the score returned by \textit{DivScore-Adaptive} is close to the true disparity score and has smaller error than \textit{DivScore-Random-Balanced}.
Unlike the collections analyzed in PPB evaluation, the attribute-specific collections of CelebA dataset are non-random; i.e., they are not i.i.d. samples from the underlying distribution.
Nevertheless, \textit{DivScore-Adaptive} is able to approximate the true disparity for each attribute-specific collection quite accurately.

Note that, for these attribute-specific collections, implementing \textit{IID-Measure} would be very expensive, since one would have to sample a small set of elements for each attribute and label them.
In comparison, our approach uses the same control set for all attributes and, hence, is much more cost-effective.

\begin{figure}[t]
    \centering
    \includegraphics[width=0.6\linewidth]{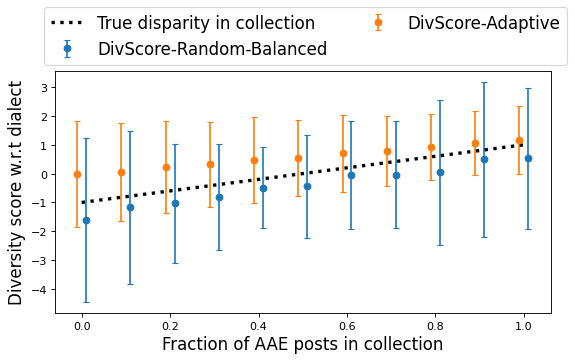}
    \caption{\small Results for TwittterAAE dataset with dialect as the protected attribute for \textit{DivScore-Random-Balanced} and \textit{DivScore-Adaptive} using control sets of size 50.}
    \label{fig:aae_results}
\end{figure}

\subsection{TwitterAAE dataset}

To show the effectiveness of \textit{DivScore} beyond image datasets, we analyze the performance on a dataset of Twitter posts.
The \textit{TwitterAAE} dataset, constructed by \citet{blodgett2016demographic}, contains around 60 million Twitter posts\footnote{\url{slanglab.cs.umass.edu/TwitterAAE/}}.
We filter the dataset to contain only posts which either are \textit{certainly} written in the African-American English (AAE) dialect (100k posts) or the White English dialect (WHE) (1.06 million posts). 
The details of filtering and feature extraction using a pre-trained Word2Vec model \cite{mikolov2015computing} are given in Appendix~\ref{sec:impl_details}.

\noindent
\paragraph{Methodology.}
For this dataset, we will evaluate the performance of \textit{DivScore-Random-Balanced} and \textit{DivScore-Adaptive}\footnoteref{note1}.
We partition the datasets into two parts: the first contains 200 posts and the second contains the rest.
The first partition is used to construct control sets of size 50 (randomly chosen from first partition for \textit{DivScore-Random-Balanced} and using Algorithm~\ref{alg:control_construction} for \textit{DivScore-Adaptive}).
The protected attribute is the dialect of the post.
The second partition is used for diversity audit evaluation.
We construct sub-datasets or collections with custom distribution of posts from each dialect.
For a given $f \in \set{0, 0.1, \dots, 1.0}$, we construct a sub-dataset $S$ of the second partition containing $f\cdot|S|$ AAE posts.
The overall size of the sampled collection is kept to be 1000 and we perform 25 repetitions.
For \textit{DivScore-Adaptive}, we use $\alpha=0.1$.

\noindent
\paragraph{Results.} The audit results for collections from TwitterAAE dataset are presented in Figure~\ref{fig:aae_results}.
The plot shows that both \textit{DivScore-Random-Balanced} and \textit{DivScore-Adaptive} can, on expectation, approximate the disparity for all collections; the disparity estimate from both methods increases with increasing fraction of AAE posts in the collection.
However, once again, the approximation error of \textit{DivScore-Adaptive} is smaller than the approximation error of \textit{DivScore-Random-Balanced} in most cases.

\section{Applications, Limitations \& Future Work}
In this section, we discuss the potential applications 
of our framework, some practical limitations and directions for future work.

\noindent
\paragraph{Third-party implementations and auditing summaries.}
To audit the diversity of any collection, \textit{DivScore} simply requires access to a small labeled control set and a similarity metric.
The cost of constructing these components is relatively small (compared to labeling the entire collection) and, hence, our audit framework can be potentially employed by third-party agencies that audit independently of the organization owning/providing the collections.
For instance, our algorithm can be implemented as a browser plugin to audit the gender diversity of Google Image results or the dialect diversity of Twitter search results.
Such a domain-generic diversity audit mechanism
can be used to ensure a more-balanced power dynamic between the organizations disseminating/controlling the data and the users of the applications that use this data.

\noindent
\paragraph{Variable-sized collections.} \textit{DivScore} can easily adapt to updates to the collections being audited.
If an element is added/removed, one simply needs to add/remove the contribution of this element from $\sm(S,T_0)$ and $\sm(S,T_1)$, and recompute $\hat{d}(S)$.
This feature crucially addresses the main drawback of \textit{IID-Measure}.

\noindent
\paragraph{Possibility of stereotype exaggeration.}
In our simulations, we evaluate the gender diversity using ``male'' vs ``female'' partition and skintone diversity using the Fitzpatrick scale.
Pre-defined protected attribute partitions, however, can be problematic; e.g., commercial AI tools' inability in handling non-binary gender \cite{scheuerman2019computers}.

Considering that 
Our algorithm is based on choosing control sets that can differentiate across protected attribute types, there is a possibility that the automatically constructed control sets can be stereotypically-biased.
For example, a control set with high $\gamma^{(T)}$ value for gender may just include images of men and women, and exclude images of transgender individuals.
While \textit{non-redundancy} aims to ensure that the control set is diverse, it does not guarantee that the control set will be perfectly representative.
Given this possibility, we strongly encourage the additional hand-curation of automatically-constructed control sets.
Further, any agency using control sets should make them public and elicit community feedback to avoid representational biases.
Recent work on designs for such cooperative frameworks can be employed for this purpose \cite{muller2020interrogating,finnegan2015counter}.

\noindent
\paragraph{Choice of $\alpha$.} For \textit{DivScore-Adaptive}, $\alpha$ is the parameter that controls the redundancy of the control set.
It primarily depends on the domain in consideration and we use fixed $\alpha$ for collections from the same domain. 
However, the mechanism to choose the best $\alpha$ for a given domain is unclear and can be further explored.

\noindent
\paragraph{Improving theoretical bounds.} While the theoretical bounds provide intuition about dependence of error on size of control set and $\gamma$, the constants in the bounds can be further improved.
E.g., in case of PPB dataset with gender protected attribute and the empirical setup in Section~\ref{sec:ppb_random_control}, Theorem~\ref{thm:main} suggests that error $|\delta| \leq 5$; however, we observe that the error is much smaller ($\leq 0.5$) in practice.
Improved and tighter analysis can help reduce the difference between theoretical and empirical performance.

\noindent
\paragraph{Assessing qualitative disparities.}
Our approach is more cost-effective than crowd annotation.
However, crowd-annotation can help answer questions about the collection beyond disparity quantification.
For example, \citet{kay2015unequal} use crowd-annotation to provide evidence of sexualized depictions of women in Google Image results for  certain occupations such as construction worker.
As part of future work, one can explore extensions of our approach or control sets that can assess such qualitative disparities as well.

\section{Conclusion}
We propose a method, \textit{DivScore}, to audit the diversity of any given collection using a small control set and an algorithm to construct adaptive control sets.
Theoretical analysis shows that \textit{DivScore} approximates the disparity of the collection, given appropriate control sets and similarity metrics.
Empirical evaluations 
demonstrate that \textit{DivScore} can handle collections from both image and text domains\footnote{Code available at \url{https://github.com/vijaykeswani/Diversity-Audit-Using-Representative-Examples}.}.

\section*{Acknowledgements}
This research was supported in part by a J.P. Morgan Faculty Award. We would like to thank Nisheeth K. Vishnoi for discussions on this problem.

\newpage
\bibliographystyle{ACM-Reference-Format}
\bibliography{references}


\begin{thebibliography}{37}


\ifx \showCODEN    \undefined \def \showCODEN     #1{\unskip}     \fi
\ifx \showDOI      \undefined \def \showDOI       #1{#1}\fi
\ifx \showISBNx    \undefined \def \showISBNx     #1{\unskip}     \fi
\ifx \showISBNxiii \undefined \def \showISBNxiii  #1{\unskip}     \fi
\ifx \showISSN     \undefined \def \showISSN      #1{\unskip}     \fi
\ifx \showLCCN     \undefined \def \showLCCN      #1{\unskip}     \fi
\ifx \shownote     \undefined \def \shownote      #1{#1}          \fi
\ifx \showarticletitle \undefined \def \showarticletitle #1{#1}   \fi
\ifx \showURL      \undefined \def \showURL       {\relax}        \fi
\providecommand\bibfield[2]{#2}
\providecommand\bibinfo[2]{#2}
\providecommand\natexlab[1]{#1}
\providecommand\showeprint[2][]{arXiv:#2}

\bibitem[\protect\citeauthoryear{Bair}{Bair}{2013}]%
        {bair2013semi}
\bibfield{author}{\bibinfo{person}{Eric Bair}.}
  \bibinfo{year}{2013}\natexlab{}.
\newblock \showarticletitle{Semi-supervised clustering methods}.
\newblock \bibinfo{journal}{\emph{Wiley Interdisciplinary Reviews:
  Computational Statistics}} \bibinfo{volume}{5}, \bibinfo{number}{5}
  (\bibinfo{year}{2013}), \bibinfo{pages}{349--361}.
\newblock


\bibitem[\protect\citeauthoryear{Blodgett, Green, and O’Connor}{Blodgett
  et~al\mbox{.}}{2016}]%
        {blodgett2016demographic}
\bibfield{author}{\bibinfo{person}{Su~Lin Blodgett}, \bibinfo{person}{Lisa
  Green}, {and} \bibinfo{person}{Brendan O’Connor}.}
  \bibinfo{year}{2016}\natexlab{}.
\newblock \showarticletitle{Demographic Dialectal Variation in Social Media: A
  Case Study of African-American English}. In
  \bibinfo{booktitle}{\emph{Proceedings of Conference on Empirical Methods in
  Natural Language Processing}}.
\newblock


\bibitem[\protect\citeauthoryear{Bolukbasi, Chang, Zou, Saligrama, and
  Kalai}{Bolukbasi et~al\mbox{.}}{2016}]%
        {bolukbasi2016man}
\bibfield{author}{\bibinfo{person}{Tolga Bolukbasi}, \bibinfo{person}{Kai-Wei
  Chang}, \bibinfo{person}{James~Y Zou}, \bibinfo{person}{Venkatesh Saligrama},
  {and} \bibinfo{person}{Adam~T Kalai}.} \bibinfo{year}{2016}\natexlab{}.
\newblock \showarticletitle{Man is to computer programmer as woman is to
  homemaker? debiasing word embeddings}. In \bibinfo{booktitle}{\emph{Advances
  in Neural Information Processing Systems}}.
\newblock


\bibitem[\protect\citeauthoryear{Buolamwini and Gebru}{Buolamwini and
  Gebru}{2018}]%
        {buolamwini2018gender}
\bibfield{author}{\bibinfo{person}{Joy Buolamwini} {and}
  \bibinfo{person}{Timnit Gebru}.} \bibinfo{year}{2018}\natexlab{}.
\newblock \showarticletitle{Gender shades: Intersectional accuracy disparities
  in commercial gender classification}. In \bibinfo{booktitle}{\emph{FAT*
  2018}}. \bibinfo{pages}{77--91}.
\newblock


\bibitem[\protect\citeauthoryear{Caliskan, Bryson, and Narayanan}{Caliskan
  et~al\mbox{.}}{2017}]%
        {caliskan2017semantics}
\bibfield{author}{\bibinfo{person}{Aylin Caliskan}, \bibinfo{person}{Joanna~J
  Bryson}, {and} \bibinfo{person}{Arvind Narayanan}.}
  \bibinfo{year}{2017}\natexlab{}.
\newblock \showarticletitle{Semantics derived automatically from language
  corpora contain human-like biases}.
\newblock \bibinfo{journal}{\emph{Science}} (\bibinfo{year}{2017}).
\newblock


\bibitem[\protect\citeauthoryear{Carbonell and Goldstein}{Carbonell and
  Goldstein}{1998}]%
        {carbonell1998use}
\bibfield{author}{\bibinfo{person}{Jaime~G Carbonell} {and}
  \bibinfo{person}{Jade Goldstein}.} \bibinfo{year}{1998}\natexlab{}.
\newblock \showarticletitle{The use of MMR, diversity-based reranking for
  reordering documents and producing summaries.}. In
  \bibinfo{booktitle}{\emph{SIGIR}}, Vol.~\bibinfo{volume}{98}.
\newblock


\bibitem[\protect\citeauthoryear{Celis and Keswani}{Celis and Keswani}{2020}]%
        {celis2020implicit}
\bibfield{author}{\bibinfo{person}{L~Elisa Celis} {and} \bibinfo{person}{Vijay
  Keswani}.} \bibinfo{year}{2020}\natexlab{}.
\newblock \showarticletitle{Implicit Diversity in Image Summarization}.
\newblock \bibinfo{journal}{\emph{Proceedings of the ACM on Human-Computer
  Interaction}} \bibinfo{volume}{4}, \bibinfo{number}{CSCW2}
  (\bibinfo{year}{2020}), \bibinfo{pages}{1--28}.
\newblock


\bibitem[\protect\citeauthoryear{Choi, Grover, Singh, Shu, and Ermon}{Choi
  et~al\mbox{.}}{2020}]%
        {choi2020fair}
\bibfield{author}{\bibinfo{person}{Kristy Choi}, \bibinfo{person}{Aditya
  Grover}, \bibinfo{person}{Trisha Singh}, \bibinfo{person}{Rui Shu}, {and}
  \bibinfo{person}{Stefano Ermon}.} \bibinfo{year}{2020}\natexlab{}.
\newblock \showarticletitle{Fair generative modeling via weak supervision}. In
  \bibinfo{booktitle}{\emph{ICML}}. PMLR.
\newblock


\bibitem[\protect\citeauthoryear{Collins}{Collins}{2002}]%
        {collins2002black}
\bibfield{author}{\bibinfo{person}{Patricia~Hill Collins}.}
  \bibinfo{year}{2002}\natexlab{}.
\newblock \bibinfo{booktitle}{\emph{Black feminist thought: Knowledge,
  consciousness, and the politics of empowerment}}.
\newblock \bibinfo{publisher}{routledge}.
\newblock


\bibitem[\protect\citeauthoryear{Davidson, Bhattacharya, and Weber}{Davidson
  et~al\mbox{.}}{2019}]%
        {davidson2019racial}
\bibfield{author}{\bibinfo{person}{Thomas Davidson}, \bibinfo{person}{Debasmita
  Bhattacharya}, {and} \bibinfo{person}{Ingmar Weber}.}
  \bibinfo{year}{2019}\natexlab{}.
\newblock \showarticletitle{Racial Bias in Hate Speech and Abusive Language
  Detection Datasets}. In \bibinfo{booktitle}{\emph{Proceedings of the Third
  Workshop on Abusive Language Online}}. \bibinfo{pages}{25--35}.
\newblock


\bibitem[\protect\citeauthoryear{Deng, Dong, Socher, Li, Li, and Fei-Fei}{Deng
  et~al\mbox{.}}{2009}]%
        {imagenet}
\bibfield{author}{\bibinfo{person}{Jia Deng}, \bibinfo{person}{Wei Dong},
  \bibinfo{person}{Richard Socher}, \bibinfo{person}{Li-Jia Li},
  \bibinfo{person}{Kai Li}, {and} \bibinfo{person}{Li Fei-Fei}.}
  \bibinfo{year}{2009}\natexlab{}.
\newblock \showarticletitle{Imagenet: A large-scale hierarchical image
  database}. In \bibinfo{booktitle}{\emph{2009 IEEE conference on computer
  vision and pattern recognition}}. Ieee, \bibinfo{pages}{248--255}.
\newblock


\bibitem[\protect\citeauthoryear{Finnegan, Oakhill, and Garnham}{Finnegan
  et~al\mbox{.}}{2015}]%
        {finnegan2015counter}
\bibfield{author}{\bibinfo{person}{Eimear Finnegan}, \bibinfo{person}{Jane
  Oakhill}, {and} \bibinfo{person}{Alan Garnham}.}
  \bibinfo{year}{2015}\natexlab{}.
\newblock \showarticletitle{Counter-stereotypical pictures as a strategy for
  overcoming spontaneous gender stereotypes}.
\newblock \bibinfo{journal}{\emph{Frontiers in psychology}}
  \bibinfo{volume}{6} (\bibinfo{year}{2015}), \bibinfo{pages}{1291}.
\newblock


\bibitem[\protect\citeauthoryear{Fitzpatrick}{Fitzpatrick}{2008}]%
        {fitzpatrick2008fitzpatrick}
\bibfield{author}{\bibinfo{person}{T Fitzpatrick}.}
  \bibinfo{year}{2008}\natexlab{}.
\newblock \showarticletitle{Fitzpatrick Skin Type Classification Scale}.
\newblock \bibinfo{journal}{\emph{Skin Inc}} (\bibinfo{year}{2008}).
\newblock


\bibitem[\protect\citeauthoryear{Gerbner, Gross, Morgan, and
  Signorielli}{Gerbner et~al\mbox{.}}{1986}]%
        {gerbner1986living}
\bibfield{author}{\bibinfo{person}{George Gerbner}, \bibinfo{person}{Larry
  Gross}, \bibinfo{person}{Michael Morgan}, {and} \bibinfo{person}{Nancy
  Signorielli}.} \bibinfo{year}{1986}\natexlab{}.
\newblock \showarticletitle{Living with television: The dynamics of the
  cultivation process}.
\newblock \bibinfo{journal}{\emph{Perspectives on media effects}}
  \bibinfo{volume}{1986} (\bibinfo{year}{1986}), \bibinfo{pages}{17--40}.
\newblock


\bibitem[\protect\citeauthoryear{Godin}{Godin}{2019}]%
        {godin2019improving}
\bibfield{author}{\bibinfo{person}{Fr{\'e}deric Godin}.}
  \bibinfo{year}{2019}\natexlab{}.
\newblock \showarticletitle{Improving and Interpreting Neural Networks for
  Word-Level Prediction Tasks in Natural Language Processing}.
\newblock \bibinfo{journal}{\emph{Ghent University}} (\bibinfo{year}{2019}).
\newblock


\bibitem[\protect\citeauthoryear{Harris}{Harris}{1982}]%
        {harris1982mammies}
\bibfield{author}{\bibinfo{person}{Trudier Harris}.}
  \bibinfo{year}{1982}\natexlab{}.
\newblock \bibinfo{booktitle}{\emph{From mammies to militants: Domestics in
  black American literature}}.
\newblock \bibinfo{publisher}{Temple University Press}.
\newblock


\bibitem[\protect\citeauthoryear{Hoeffding}{Hoeffding}{1994}]%
        {hoeffding1994probability}
\bibfield{author}{\bibinfo{person}{Wassily Hoeffding}.}
  \bibinfo{year}{1994}\natexlab{}.
\newblock \showarticletitle{Probability inequalities for sums of bounded random
  variables}.
\newblock In \bibinfo{booktitle}{\emph{The Collected Works of Wassily
  Hoeffding}}. \bibinfo{publisher}{Springer}, \bibinfo{pages}{409--426}.
\newblock


\bibitem[\protect\citeauthoryear{Kallus, Mao, and Zhou}{Kallus
  et~al\mbox{.}}{2020}]%
        {kallus2020assessing}
\bibfield{author}{\bibinfo{person}{Nathan Kallus}, \bibinfo{person}{Xiaojie
  Mao}, {and} \bibinfo{person}{Angela Zhou}.} \bibinfo{year}{2020}\natexlab{}.
\newblock \showarticletitle{Assessing algorithmic fairness with unobserved
  protected class using data combination}. In
  \bibinfo{booktitle}{\emph{Proceedings of the 2020 Conference on Fairness,
  Accountability, and Transparency}}. \bibinfo{pages}{110--110}.
\newblock


\bibitem[\protect\citeauthoryear{Kay, Matuszek, and Munson}{Kay
  et~al\mbox{.}}{2015}]%
        {kay2015unequal}
\bibfield{author}{\bibinfo{person}{Matthew Kay}, \bibinfo{person}{Cynthia
  Matuszek}, {and} \bibinfo{person}{Sean~A Munson}.}
  \bibinfo{year}{2015}\natexlab{}.
\newblock \showarticletitle{Unequal representation and gender stereotypes in
  image search results for occupations}. In
  \bibinfo{booktitle}{\emph{Proceedings of ACM Conference on Human Factors in
  Computing Systems}}.
\newblock


\bibitem[\protect\citeauthoryear{Keswani and Celis}{Keswani and Celis}{2021}]%
        {celis2021dialect}
\bibfield{author}{\bibinfo{person}{Vijay Keswani} {and}
  \bibinfo{person}{L~Elisa Celis}.} \bibinfo{year}{2021}\natexlab{}.
\newblock \showarticletitle{Dialect Diversity in Text Summarization on
  Twitter}.
\newblock \bibinfo{journal}{\emph{The Web Conference'2021}}
  (\bibinfo{year}{2021}).
\newblock


\bibitem[\protect\citeauthoryear{Kulesza, Taskar, et~al\mbox{.}}{Kulesza
  et~al\mbox{.}}{2012}]%
        {kulesza2012determinantal}
\bibfield{author}{\bibinfo{person}{Alex Kulesza}, \bibinfo{person}{Ben Taskar},
  {et~al\mbox{.}}} \bibinfo{year}{2012}\natexlab{}.
\newblock \showarticletitle{Determinantal point processes for machine
  learning}.
\newblock \bibinfo{journal}{\emph{Foundations and Trends{\textregistered} in
  Machine Learning}} \bibinfo{volume}{5}, \bibinfo{number}{2--3}
  (\bibinfo{year}{2012}), \bibinfo{pages}{123--286}.
\newblock


\bibitem[\protect\citeauthoryear{Larrazabal, Nieto, Peterson, Milone, and
  Ferrante}{Larrazabal et~al\mbox{.}}{2020}]%
        {larrazabal2020gender}
\bibfield{author}{\bibinfo{person}{Agostina~J Larrazabal},
  \bibinfo{person}{Nicol{\'a}s Nieto}, \bibinfo{person}{Victoria Peterson},
  \bibinfo{person}{Diego~H Milone}, {and} \bibinfo{person}{Enzo Ferrante}.}
  \bibinfo{year}{2020}\natexlab{}.
\newblock \showarticletitle{Gender imbalance in medical imaging datasets
  produces biased classifiers for computer-aided diagnosis}.
\newblock \bibinfo{journal}{\emph{Proceedings of the National Academy of
  Sciences}} \bibinfo{volume}{117}, \bibinfo{number}{23}
  (\bibinfo{year}{2020}), \bibinfo{pages}{12592--12594}.
\newblock


\bibitem[\protect\citeauthoryear{Liu, Luo, Wang, and Tang}{Liu
  et~al\mbox{.}}{2015}]%
        {liu2015faceattributes}
\bibfield{author}{\bibinfo{person}{Ziwei Liu}, \bibinfo{person}{Ping Luo},
  \bibinfo{person}{Xiaogang Wang}, {and} \bibinfo{person}{Xiaoou Tang}.}
  \bibinfo{year}{2015}\natexlab{}.
\newblock \showarticletitle{Deep Learning Face Attributes in the Wild}. In
  \bibinfo{booktitle}{\emph{ICCV 2015}}.
\newblock


\bibitem[\protect\citeauthoryear{Mikolov, Chen, Corrado, and Dean}{Mikolov
  et~al\mbox{.}}{2015}]%
        {mikolov2015computing}
\bibfield{author}{\bibinfo{person}{Tomas Mikolov}, \bibinfo{person}{Kai Chen},
  \bibinfo{person}{Gregory~S Corrado}, {and} \bibinfo{person}{Jeffrey~A Dean}.}
  \bibinfo{year}{2015}\natexlab{}.
\newblock \bibinfo{title}{Computing numeric representations of words in a
  high-dimensional space}.
\newblock
\newblock
\newblock
\shownote{US Patent 9,037,464.}


\bibitem[\protect\citeauthoryear{Mitzenmacher and Upfal}{Mitzenmacher and
  Upfal}{2017}]%
        {mitzenmacher2017probability}
\bibfield{author}{\bibinfo{person}{Michael Mitzenmacher} {and}
  \bibinfo{person}{Eli Upfal}.} \bibinfo{year}{2017}\natexlab{}.
\newblock \bibinfo{booktitle}{\emph{Probability and computing: Randomization
  and probabilistic techniques in algorithms and data analysis}}.
\newblock


\bibitem[\protect\citeauthoryear{Muller, Aragon, Guha, Kogan, Neff, Seidelin,
  Shilton, and Tanweer}{Muller et~al\mbox{.}}{2020}]%
        {muller2020interrogating}
\bibfield{author}{\bibinfo{person}{Michael Muller}, \bibinfo{person}{Cecilia
  Aragon}, \bibinfo{person}{Shion Guha}, \bibinfo{person}{Marina Kogan},
  \bibinfo{person}{Gina Neff}, \bibinfo{person}{Cathrine Seidelin},
  \bibinfo{person}{Katie Shilton}, {and} \bibinfo{person}{Anissa Tanweer}.}
  \bibinfo{year}{2020}\natexlab{}.
\newblock \showarticletitle{Interrogating Data Science}. In
  \bibinfo{booktitle}{\emph{Conference Companion Publication of the 2020 on
  CSCW}}. \bibinfo{pages}{467--473}.
\newblock


\bibitem[\protect\citeauthoryear{Nguyen and Bai}{Nguyen and Bai}{2010}]%
        {nguyen2010cosine}
\bibfield{author}{\bibinfo{person}{Hieu~V Nguyen} {and} \bibinfo{person}{Li
  Bai}.} \bibinfo{year}{2010}\natexlab{}.
\newblock \showarticletitle{Cosine similarity metric learning for face
  verification}. In \bibinfo{booktitle}{\emph{Asian conference on computer
  vision}}. Springer, \bibinfo{pages}{709--720}.
\newblock


\bibitem[\protect\citeauthoryear{Noble}{Noble}{2018}]%
        {noble2018algorithms}
\bibfield{author}{\bibinfo{person}{Safiya~Umoja Noble}.}
  \bibinfo{year}{2018}\natexlab{}.
\newblock \bibinfo{booktitle}{\emph{Algorithms of oppression: How search
  engines reinforce racism}}.
\newblock \bibinfo{publisher}{nyu Press}.
\newblock


\bibitem[\protect\citeauthoryear{O'neil}{O'neil}{2016}]%
        {o2016weapons}
\bibfield{author}{\bibinfo{person}{Cathy O'neil}.}
  \bibinfo{year}{2016}\natexlab{}.
\newblock \bibinfo{booktitle}{\emph{Weapons of math destruction: How big data
  increases inequality and threatens democracy}}.
\newblock \bibinfo{publisher}{Broadway Books}.
\newblock


\bibitem[\protect\citeauthoryear{Scheuerman, Paul, and Brubaker}{Scheuerman
  et~al\mbox{.}}{2019}]%
        {scheuerman2019computers}
\bibfield{author}{\bibinfo{person}{Morgan~Klaus Scheuerman},
  \bibinfo{person}{Jacob~M Paul}, {and} \bibinfo{person}{Jed~R Brubaker}.}
  \bibinfo{year}{2019}\natexlab{}.
\newblock \showarticletitle{How Computers See Gender: An Evaluation of Gender
  Classification in Commercial Facial Analysis Services}.
\newblock \bibinfo{journal}{\emph{Proceedings of ACM Human-Computer
  Interaction}} (\bibinfo{year}{2019}).
\newblock


\bibitem[\protect\citeauthoryear{Schwemmer, Knight, Bello-Pardo, Oklobdzija,
  Schoonvelde, and Lockhart}{Schwemmer et~al\mbox{.}}{2020}]%
        {schwemmer2020diagnosing}
\bibfield{author}{\bibinfo{person}{Carsten Schwemmer}, \bibinfo{person}{Carly
  Knight}, \bibinfo{person}{Emily~D Bello-Pardo}, \bibinfo{person}{Stan
  Oklobdzija}, \bibinfo{person}{Martijn Schoonvelde}, {and}
  \bibinfo{person}{Jeffrey~W Lockhart}.} \bibinfo{year}{2020}\natexlab{}.
\newblock \showarticletitle{Diagnosing gender bias in image recognition
  systems}.
\newblock \bibinfo{journal}{\emph{Socius}}  \bibinfo{volume}{6}
  (\bibinfo{year}{2020}), \bibinfo{pages}{2378023120967171}.
\newblock


\bibitem[\protect\citeauthoryear{Shrum}{Shrum}{1995}]%
        {shrum1995assessing}
\bibfield{author}{\bibinfo{person}{Larry~J Shrum}.}
  \bibinfo{year}{1995}\natexlab{}.
\newblock \showarticletitle{Assessing the social influence of television: A
  social cognition perspective on cultivation effects}.
\newblock \bibinfo{journal}{\emph{Communication Research}}
  \bibinfo{volume}{22}, \bibinfo{number}{4} (\bibinfo{year}{1995}).
\newblock


\bibitem[\protect\citeauthoryear{Simonyan and Zisserman}{Simonyan and
  Zisserman}{2014}]%
        {simonyan2014very}
\bibfield{author}{\bibinfo{person}{Karen Simonyan} {and}
  \bibinfo{person}{Andrew Zisserman}.} \bibinfo{year}{2014}\natexlab{}.
\newblock \showarticletitle{Very deep convolutional networks for large-scale
  image recognition}.
\newblock \bibinfo{journal}{\emph{arXiv preprint arXiv:1409.1556}}
  (\bibinfo{year}{2014}).
\newblock


\bibitem[\protect\citeauthoryear{Singh, Chayko, Inamdar, and Floegel}{Singh
  et~al\mbox{.}}{2020}]%
        {singh2019female}
\bibfield{author}{\bibinfo{person}{Vivek~K Singh}, \bibinfo{person}{Mary
  Chayko}, \bibinfo{person}{Raj Inamdar}, {and} \bibinfo{person}{Diana
  Floegel}.} \bibinfo{year}{2020}\natexlab{}.
\newblock \showarticletitle{Female Librarians and Male Computer Programmers?
  Gender Bias in Occupational Images on Digital Media Platforms}.
\newblock \bibinfo{journal}{\emph{Journal of the Association for Information
  Science and Technology}} (\bibinfo{year}{2020}).
\newblock


\bibitem[\protect\citeauthoryear{Word, Zanna, and Cooper}{Word
  et~al\mbox{.}}{1974}]%
        {word1974nonverbal}
\bibfield{author}{\bibinfo{person}{Carl~O Word}, \bibinfo{person}{Mark~P
  Zanna}, {and} \bibinfo{person}{Joel Cooper}.}
  \bibinfo{year}{1974}\natexlab{}.
\newblock \showarticletitle{The nonverbal mediation of self-fulfilling
  prophecies in interracial interaction}.
\newblock \bibinfo{journal}{\emph{Journal of experimental social psychology}}
  \bibinfo{volume}{10}, \bibinfo{number}{2} (\bibinfo{year}{1974}),
  \bibinfo{pages}{109--120}.
\newblock


\bibitem[\protect\citeauthoryear{Xing, Jordan, Russell, and Ng}{Xing
  et~al\mbox{.}}{2003}]%
        {xing2003distance}
\bibfield{author}{\bibinfo{person}{Eric~P Xing}, \bibinfo{person}{Michael~I
  Jordan}, \bibinfo{person}{Stuart~J Russell}, {and} \bibinfo{person}{Andrew~Y
  Ng}.} \bibinfo{year}{2003}\natexlab{}.
\newblock \showarticletitle{Distance metric learning with application to
  clustering with side-information}. In \bibinfo{booktitle}{\emph{NIPS}}.
\newblock


\bibitem[\protect\citeauthoryear{Zhu and Goldberg}{Zhu and Goldberg}{2009}]%
        {zhu2009introduction}
\bibfield{author}{\bibinfo{person}{Xiaojin Zhu} {and} \bibinfo{person}{Andrew~B
  Goldberg}.} \bibinfo{year}{2009}\natexlab{}.
\newblock \showarticletitle{Introduction to semi-supervised learning}.
\newblock \bibinfo{journal}{\emph{Synthesis lectures on artificial intelligence
  and machine learning}} (\bibinfo{year}{2009}).
\newblock


\end{thebibliography}

\newpage
\appendix

\section{Proofs}
\subsection{Proof of Lemma~\ref{lem:single_elem}} \label{sec:proofs}
\begin{proof}[Proof of Lemma~\ref{lem:single_elem}]
Suppose $x$ has protected attribute type 0, i.e., $x \in S_0$. 
Since control set $T$ is domain-relevant, we know that for any $y \in T_0$, $\E[\sm(x,y)] = \ms$ and for any $y \in T_1$, $\E[\sm(x,y)] = \md$.
Then, using Chernoff-Hoeffding bounds \cite{hoeffding1994probability,mitzenmacher2017probability}, we get that for any $\delta > 0$,
\[\P\left[ \sm(x, T_0) \notin (1 \pm \delta)\ms \right] \leq 2\exp(-\delta^2 \cdot |T_0| \cdot \ms/3), \textrm{ and }\]
\[\P\left[ \sm(x, T_1) \notin (1 \pm \delta)\md \right] \leq  2\exp(-\delta^2 \cdot |T_1| \cdot \md/3).\]
Note that $|T_0| = |T_1| = |T|/2$. The probability that both the above events are simultaneously satisfied is
\begin{align*}
2 \exp(-\delta^2 &  \ms|T| /6) + 2\exp(-\delta^2  \md|T| /6) \\ &\leq 2\exp(-\delta^2 \md |T| /6) \cdot \left(1 + \exp(-\delta^2 \gamma |T| /6) \right).    
\end{align*}

\noindent
Therefore, combining the two statements we get that with probability atleast $1 - 2\exp(-\delta^2 \md |T| /6) \cdot \left(1 + \exp(-\delta^2 \gamma |T| /6) \right)$,
\begin{align*}
&\sm(x, T_0) - \sm(x,T_1)\\ &\in \left[ (1-\delta)\ms - (1+\delta)\md, (1+\delta)\ms - (1-\delta)\md \right]).
\end{align*}
Simplifying the above expression, we get
\[\sm(x, T_i) - \sm(x,T_{1-i}) \in  \ms - \md \pm \delta (\ms + \md).\]
The other direction (when $x \in S_1$) follows from symmetry.
\end{proof}

\subsection{Proof of Theorem~\ref{thm:main}} \label{sec:proofs_2}
\begin{proof}[Proof of Theorem~\ref{thm:main}]
%
Applying Lemma~\ref{lem:single_elem} to each element in $S$, we get that with probability atleast $q := 1 -  2|S|e^{-\delta^2 \md |T| /6} \cdot (1 + e^{-\delta^2 \gamma |T| /6} )$, all elements satisfy condition \eqref{1}.
Summing up $\sm(x, T_0) - \sm(x,T_1)$ for all $x \in S$, we get
\begin{align*}
\sm(S, T_0) - & \sm(S,T_1)\\  &\in (\ms - \md) \cdot \left(|S_0 |- |S_1|\right)/|S| \pm \delta (\ms + \md).    
\end{align*}
Simplifying the above bound, we have that with probability $q$,
\[\hat{d}(S) \in(\ms - \md) \cdot  d(S) \pm \delta \cdot (\ms + \md).\]
%
%

%
\noindent
By choosing $\delta = \sqrt{\frac{6\log (20|S| )}{ |T| \min(\md, \lambda)}}$, the probability $q$ is atleast 
\begin{align*}
1 - & 2|S|e^{-\delta^2 \md |T| /6} (1 + e^{-\delta^2 \gamma |T| /6} ) \\ &\geq 1 - 2|S| e^{-\log 20|S|}  (1 +  e^{-\log 20|S|})= 0.9- \frac{1}{200|S|}.    
\end{align*}
\end{proof}

\section{Implementation Details} \label{sec:impl_details}

\noindent
\paragraph{Details of \textit{SS-ST} baseline.} \label{sec:ss_st}
The complete implementation of the semi-supervised self-training baseline \textit{SS-ST} is given in Algorithm~\ref{alg:ss_st}.
We use $k=5$ for PPB-2017 simulations.

\begin{algorithm}[!htbp]
\caption{\textit{SS-ST} baseline}
\label{alg:ss_st}
\begin{flushleft}
    \textbf{Input:} Dataset $S$, control set $T := T_0 \cup T_1$, $\sm(\cdot, \cdot)$, $k \in \mathbb{Z}_{>0}$
\end{flushleft}
\begin{algorithmic}[1]
\State $n_0, n_1 \gets 0$
\While {$S \neq \emptyset$} 
    \For{$x \in S$}
        \State $s(x) \gets \frac{1}{|T_0|} \sum_{y \in T_0} \sm(x,y) - \frac{1}{|T_1|} \sum_{y \in T_1} \sm(x,y)$
    \EndFor
    \State $\tilde{T} \gets $ top $k$ elements in set $\set{|s(x)|}_{x \in S}$
    \State $n_0 \gets n_0 + |\set{s(x) \mid x \in \tilde{T}, s(x) > 0}|$
    \State $n_1 \gets n_1 + |\set{s(x) \mid x \in \tilde{T}, s(x) < 0}|$
    \State $S \gets S \setminus \tilde{T}, T \gets T \cup \tilde{T}$
\EndWhile
\State \textbf{return} $(n_0 - n_1)/|S|$ 
\end{algorithmic}
\end{algorithm}

\noindent
\paragraph{PPB-2017 and CelebA datasets.}
For both PPB-2017 and CelebA datasets, feature extraction for images is done using the pre-trained VGG-16 deep network \cite{simonyan2014very}.
%
The network has been pre-trained on the Imagenet \cite{imagenet} dataset.
To extract the feature of any given image, we pass it as input to the network and extract the 4096-dimensional weight vector of the last fully connected layer.
We further reduce the feature vector size to 300 by performing PCA on the set of features of all images in the dataset.

\noindent
\paragraph{TwitterAAE dataset.}
For the TwitterAAE dataset, the authors constructed a demographic language identification model to report the probability of each post being written by a user of any of the following population categories:  non-Hispanic Whites, non-Hispanic Blacks, Hispanics, and Asians.
We filter the dataset to contain only posts for which probability of belonging to non-Hispanic African-American English language model or non-Hispanic White English language model is  $\geq 0.99$. 
This leads to a dataset of around 1.2 million tweets, with around 100k posts belonging to non-Hispanic African-American English language model and 1.06 million posts belonging to non-Hispanic White English language model; we will refer to the two groups of posts as AAE and WHE posts in.

To extract feature vectors corresponding to the Twitter posts, we use a Word2Vec model \cite{mikolov2015computing} pre-trained on 400 million Twitter posts \cite{godin2019improving}.
For any given post, we first use the Word2Vec model to extract features for every word in the post.
Then we take the average of the word features to obtain the feature of the post.

\begin{figure*}[t]
    \centering
    \includegraphics[width=0.9\linewidth]{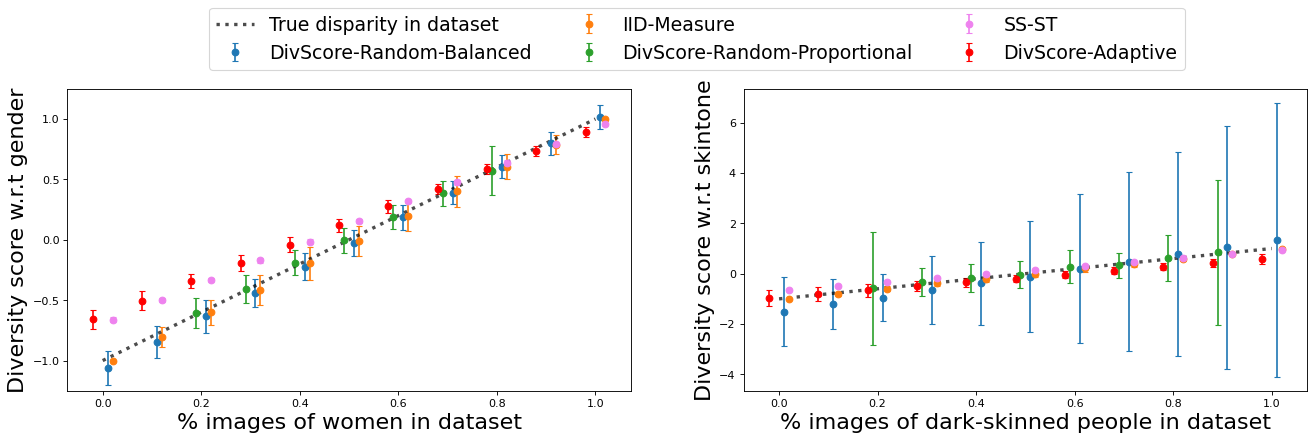}
    \subfloat[Gender protected attribute]{\hspace{.5\linewidth}}
    \subfloat[Skin-tone protected attribute]{\hspace{.5\linewidth}}
    \caption{Results for PPB-2017 dataset using random and adaptive control sets. The plots in this figure are the same as the plots in Figure~\ref{fig:ppb_results_all}, except that we don't put y-axis limitations here to present the complete errorbars for all methods.
    }
    \label{fig:ppb_results_all_alt}
\end{figure*}

\begin{figure*}[t]
    \centering
    \includegraphics[width=0.9\linewidth]{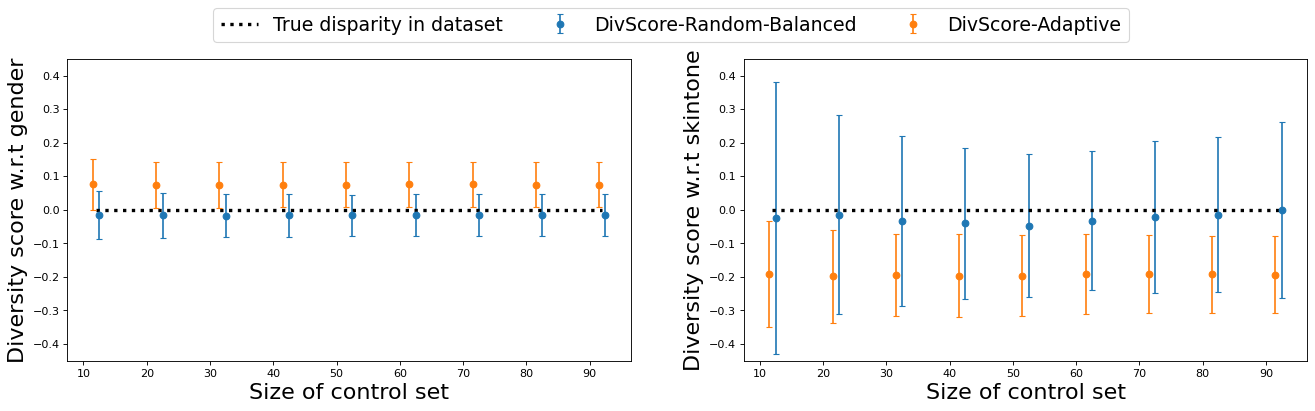}
    \subfloat[Gender protected attribute]{\hspace{.5\linewidth}}
    \subfloat[Skin-tone protected attribute]{\hspace{.5\linewidth}}
    \caption{Results for PPB-2017 dataset using different sized random and adaptive control sets. 
    }
    \label{fig:ppb_diff_control}
\end{figure*}

\begin{figure*}[t]
    \centering
    \includegraphics[width=\linewidth]{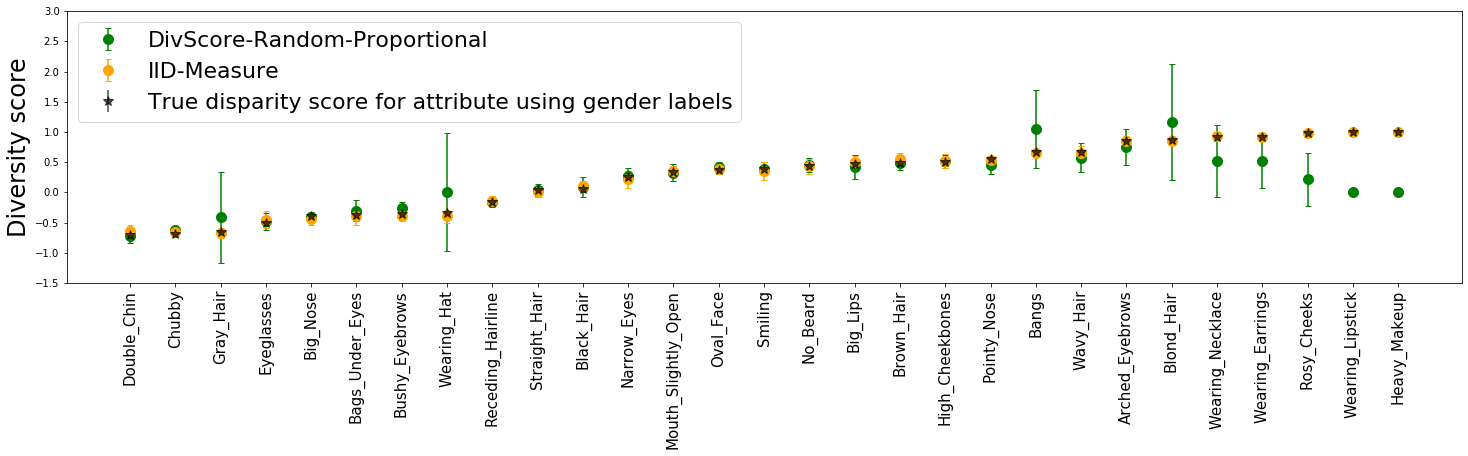}
    \caption{Performance of \textit{DivScore-Random-Proportional} and \textit{IID-Measure} on CelebA dataset}
    \label{fig:celeba_baselines}
\end{figure*}

\begin{figure*}[!ht]
    \centering
    \includegraphics[width=0.5\linewidth]{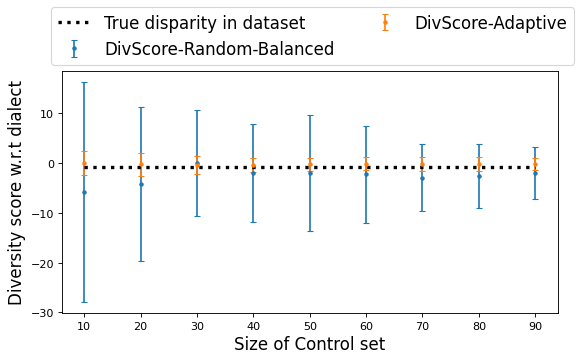}
    \caption{Results for TwitterAAE dataset using different sized random and adaptive control sets. 
    }
    \label{fig:aae_diff_control}
\end{figure*}

\newpage

\section{Other Empirical Results} \label{sec:other_results}

\noindent
\paragraph{Alternate Figure~\ref{fig:ppb_results_all} plot.} First, we present the plots from Figure~\ref{fig:ppb_results_all} without y-axis limitations. This is presented in Figure~\ref{fig:ppb_results_all_alt}.

\noindent
\paragraph{Variation of performance with control set size for PPB-dataset.} Figure~\ref{fig:ppb_diff_control} presents the variation of disparity measure with control set size.
The disparity in the collection is fixed to be 0.
The plots show that \textit{DivScore-Adaptive} can achieve low approximation error using smaller sized control sets than \textit{DivScore-Random-Balanced}.

\noindent
\paragraph{Performance of \textit{DivScore-Random-Proportional} and \textit{IID-Measure} on CelebA dataset.}
Figure~\ref{fig:celeba_baselines} presents the performance of \textit{DivScore-Random-Proportional} and \textit{IID-Measure} on for different facial attributes of CelebA dataset.
As expected, \textit{IID-Measure} has low approximation error, while \textit{DivScore-Random-Proportional} has low approximation error for some attributes and high error for others.
Nevertheless, as discussed in Section~\ref{sec:ppb_discussion}, both baselines need different control sets for collections corresponding to different attributes, and hence, are costly when auditing multiple collections from the same domain.

\noindent
\paragraph{Variation of performance with control set size for TwitterAAE-dataset.} Figure~\ref{fig:aae_diff_control} presents the variation of disparity measure with control set size.
The disparity in the collection is fixed to be -0.826 (which is the disparity of the overall dataset)
The plots show that, once again, \textit{DivScore-Adaptive} can achieve low approximation error using much smaller sized control sets than \textit{DivScore-Random-Balanced}.

\end{document}